\documentclass[runningheads]{llncs}

\usepackage[T1]{fontenc}
\usepackage[utf8]{inputenc}
\usepackage[english]{babel}
\usepackage{graphicx} 
\usepackage{mathtools}
\usepackage{amssymb}
\usepackage[linesnumbered, ruled, vlined]{algorithm2e}
\usepackage{complexity}
\usepackage{enumitem}
\usepackage[hidelinks]{hyperref}
\usepackage[usenames, dvipsnames]{xcolor}
\usepackage[caption=false]{subfig}
\usepackage{footmisc}

\SetNlSty{textbf}{\color{black}}{}

\SetCommentSty{mycommentfont}

\usepackage{etoolbox}
\apptocmd{\sloppy}{\hbadness 10000\relax}{}{}

\usepackage{tikz}
\usetikzlibrary{decorations.pathreplacing,calc}

\newcommand{\phantspace}{\ensuremath{\vphantom{|\Delta_i|}}}

\title{Output-Sensitive Enumeration of Potential Maximal Cliques in Polynomial Space}
\titlerunning{PMC Enumeration in Polynomial Space}
\author{Caroline Brosse\inst{1, 3} \and
Alessio Conte\inst{2}\orcidID{0000-0003-0770-2235} \and
Vincent Limouzy\inst{3}\orcidID{0000-0002-9133-7009} \and 
Giulia Punzi\inst{2,4}\orcidID{0000-0001-8738-1595} \and
Davide Rucci\inst{2}\orcidID{0000-0003-1273-2770}
}
\authorrunning{C. Brosse et al.}
\institute{{CNRS, Université C\^{o}te d’Azur, Inria, I3S, Sophia-Antipolis, France} \\
\email{caroline.brosse@inria.fr}\and
University of Pisa, Italy
\email{\{alessio.conte, giulia.punzi\}@unipi.it,  davide.rucci@phd.unipi.it} \and
{Universit\'e Clermont Auvergne, Clermont
              Auvergne INP, CNRS, Mines Saint-Etienne, \textsc{Limos}, F-63000
              Clermont-Ferrand, France}
\\
\email{vincent.limouzy@uca.fr} \and 
National Institute of Informatics, Japan
}

\begin{document}

\maketitle
\begin{abstract}

A set of vertices in a graph forms a \emph{potential maximal clique} 
if there exists a minimal chordal completion in which it 
is a maximal clique.
Potential maximal cliques were first introduced as a key tool 
to obtain an efficient, though exponential-time algorithm to compute the treewidth of a graph.
As a byproduct, this allowed 
to compute the treewidth of various graph classes in polynomial time.

In recent years, the concept of potential maximal cliques regained 
interest as it proved to be useful for a handful of graph algorithmic problems.
In particular, it turned out to be a key tool to obtain a polynomial time 
algorithm for computing maximum weight independent sets 
in $P_5$-free and $P_6$-free graphs (Lokshtanov et al., SODA `14 \cite{LokshantovVV14}
and Grzeskik et al., SODA `19 \cite{GrzesikKPP19}).
In most of their applications, obtaining all the potential 
maximal cliques constitutes an algorithmic bottleneck, thus motivating the question of how to efficiently enumerate all the potential maximal 
cliques in a graph $G$.

The state-of-the-art algorithm by Bouchitté \& Todinca can enumerate potential maximal cliques in output-polynomial time by using exponential space, a significant limitation for the size of feasible instances.
In this paper, we revisit this algorithm and design an enumeration algorithm that preserves an output-polynomial time complexity while only requiring polynomial space.

\end{abstract}

\clearpage

\section{Introduction}

Potential maximal cliques are fascinating objects in graph theory.
A \emph{potential maximal clique}, or \emph{PMC} for short, of a graph is a set of vertices inducing a maximal clique in some minimal triangulation of a graph (see Fig.~\ref{fig:pmc-example} for an example).
These objects were originally introduced by Bouchitté and Todinca in the late 1990s \cite{bouchitte1998minimal} as a key tool to handle the computation of treewidth and minimum fill-in of a graph, which are both \NP-complete problems.
To deal with these problems, the authors adopted an enumerative approach: their idea is to compute the list of all PMCs of the input graph before processing them as quickly as possible.
Just like for maximal cliques, the number of potential maximal cliques of a graph can be large, typically exponential in the number of vertices.
Therefore, an algorithm enumerating the potential maximal cliques of a graph cannot in general run in polynomial time.
However, the algorithm proposed by Bouchitté and Todinca can be used to build efficient exact exponential algorithms for the \NP-complete problems they considered.
\begin{figure}[ht]
    \centering
    \subfloat[]{\label{fig:pmc-example:sub:G}
    \centering
    \includegraphics[width=.32\textwidth]{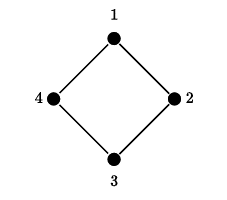}
    }
    \hfill
    \subfloat[]{\label{fig:pmc-example:sub:triang}
    \centering
    \includegraphics[width=.32\textwidth]{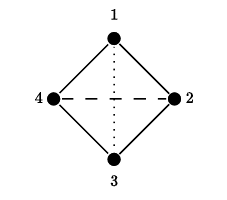}
    }
    \hfill
    \subfloat[]{\label{fig:pmc-example:sub:min-sep}
        \centering
        \includegraphics[width=.33\textwidth]{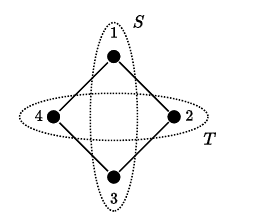}
    }
    \caption{(a) a simple graph on 4 vertices.
    (b) Two possible minimal triangulations of the graph (dashed line or dotted line).
    (c) Two minimal separators of the graph: $\{1,3\}$ and $\{2,4\}$.
    The potential maximal cliques of (a) include all the maximal cliques of all the different minimal triangulations, that is to say, $\{1,2,3\}$ and $\{1,3,4\}$ from the dotted triangulation and $\{1,2,4\}$, $\{2,3,4\}$ from the dashed triangulation.}
    \label{fig:pmc-example}
\end{figure}
The guarantee of efficiency is given by the \emph{output-polynomial} complexity measure.
An enumeration algorithm is said to have output-polynomial time complexity if its running time is polynomial both in the input size (as is usually the case) and in the number of solutions that have to be returned.
In some sense, it allows to capture a notion of running time ``per solution'' rather than simply bounding the total execution time by an exponential function.
The output-polynomial complexity also guarantees a running time that is polynomial in the input size only when the number of objects to be generated is known to be polynomial, as it is the case for PMCs, or equivalently for \emph{minimal separators}, in some classes of graphs \cite{bouchitte1998minimal}.

Due to their links with tree-decompositions and the computation of graph parameters such as treewidth, being able to list the potential maximal cliques of a graph efficiently is still at the heart of several algorithms, including recent ones \cite{ChudnovskyPPT20,FominTV15,GrzesikKPP19,LiedloffMT19,LokshantovVV14}.
Being able to guarantee the fastest possible complexity for the enumeration of PMCs is thus of great interest.
The currently best known algorithm is the one proposed by Bouchitté and Todinca \cite{bouchitte2002listing}, which is quadratic in the number of PMCs.
In their series of papers, they asked whether a linear dependency can be achieved, but this problem is still open.
On the other hand, space usage is also an issue: their algorithm needs to remember all the (exponentially many) solutions that have already been found. As memory is in practice limited this can make the computation unfeasible, thus seeking a faster algorithm without improving the space usage may not lead to practical benefits.

Despite growing interest on PMC enumeration, few advances have been made on improving the complexity of the problem \cite{korhonen2019enumerating,korhonen2019solving}.
On a closely related topic, one can mention the computation of all minimal triangulations of a graph, linked with its tree-decompositions, which has received at least two algorithmic improvements during the past decade \cite{carmeli2021efficiently,brosse2022polynomial}.
However it is still unknown if an improvement on the enumeration of all minimal triangulations can help to list the PMCs more efficiently.

\paragraph{Our contribution.}
The main result of the current paper is an algorithm that generates all the potential maximal cliques of a graph in polynomial space, while keeping the output-polynomial time complexity status.
First, in Section~\ref{sec:non-duplication} we show how to modify the original algorithm \cite{bouchitte2002listing} to avoid duplicates: we refine the tests inside the algorithm to be able not to explore the same solution twice.
Then in Section~\ref{sec:polyspace}, we reduce the space usage of our version of Bouchitté and Todinca's algorithm by modifying the exploration strategy.
The main effect of this modification is that remembering all the solutions already found is no longer needed.
Of course, the execution time is affected, but in the end our algorithm still has output-polynomial complexity, and uses only polynomial space.

\section{Definitions and Theoretical Results}

\subsection{Notations and basic concepts}

Throughout the paper, a graph will be denoted by $G=(V,E)$.
As usual, we use $n$ for the number of vertices ($n=|V|$) and $m$ for the number of edges ($m=|E|$).
The neighborhood of a vertex $u$ is the set $N(u)=\{v\in V\mid uv\in E\}$, and the neighborhood of a set of vertices $S$ is the set of vertices that have a neighbor in $S$: $N(S)=\{v\in V\setminus S\mid \exists u\in S,\ uv\in E\}$.

A \emph{minimal triangulation} of the graph $G=(V,E)$ is a chordal graph $H=(V,E\cup F)$ -- that is, a graph without induced cycles of length $4$ or more -- such that for any proper subset $F'$ of $F$, the graph $H'=(V,E\cup F')$ is not chordal.
The \emph{potential maximal cliques}, or \emph{PMC}s, of the graph $G$ are the sets of vertices inducing an inclusion-wise maximal clique in some minimal triangulation of $G$ (see Fig.~\ref{fig:pmc-example:sub:triang}).
The set of all PMCs of the graph $G$ is denoted by $\Pi_G$.

As highlighted by Bouchitté and Todinca \cite{bouchitte1998minimal}, the potential maximal cliques are closely related to other structures in graphs: the \emph{minimal separators}.
In the graph $G$, for two vertices $a$ and $b$, a \emph{minimal $(a,b)$ separator} is a set $S$ of vertices of $G$ such that $a$ and $b$ that are in distinct connected components of $G\setminus S$, and that is minimal for this property.
The \emph{minimal separators} of $G$ are all the minimal $(a,b)$ separators for all pairs $(a,b)$ of vertices (see Fig.~\ref{fig:pmc-example:sub:min-sep}).
We will use $\Delta_G$ to denote the set of all minimal separators of $G$.
The minimal separators are a crucial ingredient to build new PMCs; this is why we should be able to generate them efficiently as well.

Finally, for algorithmic purposes, we consider the vertices of a graph $G$ to be arbitrarily ordered, that is to say, $V=\{v_1,\ldots,v_n\}$.
Then, for any $1\leq i\leq n$, we define the graph $G_i$ as the subgraph of $G$ induced by the vertex set $\{v_1,\ldots,v_i\}$.
Therefore, $G_1 = (\{v_1\},\varnothing)$ and for any $i<n$, $G_{i+1}$ contains exactly one more vertex than $G_i$, together with all the edges of $G$ between it and vertices of $G_i$.
For simplicity of notation, we will use $\Pi_i$ and $\Delta_i$ instead of $\Pi_{G_i}$ and $\Delta_{G_i}$ to denote the sets of PMCs and minimal separators of $G_i$.

\subsection{Background}

To enumerate the potential maximal cliques of a graph, we base our algorithm on the first one proposed by Bouchitté and Todinca.
It uses an incremental approach: the principle is to add the vertices one by one and generate at each step $i$ the PMCs of the graph $G_i$ {by extending the PMCs of $G_{i-1}$ and generating new ones from minimal separators of $G_i$}.

As our algorithm builds up on the original one given by Bouchitté and Todinca \cite{bouchitte2002listing}, we rely on their proofs and terminology.
Namely, to decide efficiently if a set of vertices is a PMC, we will need the notion of \emph{full component}.
Initially introduced for minimal separators, this definition can in fact be given for any subset of the vertex set, even if its removal leaves the graph connected.

\begin{definition}[Full Component] Given a set $S$ of vertices of the graph $G$, a connected component $C$ of $G\setminus S$ is \emph{full for $S$} if every $v \in S$ has a neighbor in $C$.
That is to say, for any $v\in S$ there exists $u\in C$ such that $uv\in E(G)$.
\end{definition}

The potential maximal cliques have been characterized by Bouchitté and Todinca using the full components \cite[Theorem~3.15]{bouchitte2001treewidth}.
This characterization, that we state thereafter, is very useful from an algorithmic point of view since it provides a polynomial test for a subset of the vertices of a graph to be a PMC.

\begin{theorem}[Characterization of PMCs \cite{bouchitte2001treewidth}]\label{thm:characterization_pmc}
Given a graph $G=(V,E)$, a subset $K \subseteq V$ is a PMC of $G$ if and only if
\begin{enumerate}[label=(\alph*)]\itemsep0em
    \item there is no full component for $K$ in $G$, \label{pmc_char_1}
    \item for any two vertices $x$ and $y$ of $K$ such that $xy\notin E$, there exists a connected component $C$ of $G\setminus K$ such that $x, y\in N(C)$. \label{pmc_char_2}
\end{enumerate}
\end{theorem}

Theorem~\ref{thm:characterization_pmc} provides an efficient test to decide if a given subset of the vertices of $G$ is a PMC.
The test can be run in time $O(mn)$.

In our algorithm, as in the original one by Bouchitté and Todinca, the potential maximal cliques will be computed incrementally.
Bouchitté and Todinca proved that for any vertex $a$ of a graph $G$, the PMCs of $G$ can be obtained either from PMCs in $G\setminus\{a\}$, or from minimal separators of $G$ \cite[Theorem~20]{bouchitte2002listing}.
However, they did not prove that this condition was sufficient.
From Theorem~\ref{thm:characterization_pmc}, we are able to deduce the following property of potential maximal cliques: any PMC of $G\setminus\{a\}$ can be {uniquely} extended to a PMC of $G$.

\begin{proposition}\label{prop:pmc-extension}
    Let $G$ be a graph and $a$ be any vertex of $G$.
    For any PMC $K$ of $G\setminus\{a\}$, exactly one between $K$ and $K\cup \{a\}$ is a PMC of $G$.
\end{proposition}

\begin{proof}
    We define $G':=G\setminus\{a\}$, so that $G=G'+a$.
    Let $K$ be a PMC of $G'$ and suppose that $K$ is not a PMC of $G$.
    
    We start by proving that condition~\ref{pmc_char_2} of Theorem \ref{thm:characterization_pmc} remains true for $K$ in $G$.
    Let $x$ and $y \in K$ such that $xy\notin E$ (if two such vertices exist, otherwise \ref{pmc_char_2} is true by emptiness).
    Since $K$ is a PMC of $G'$, there exists a connected component $C'$ of $G'\setminus K$ such that $x$ and $y$ both have neighbors in $C'$.
    We know that any connected component of $G'\setminus K$ is contained in some connected component of $G \setminus K$; in particular there exists a connected component $C$ of $G\setminus K$ such that $C'\subseteq C$.
    Consequently, $x$ and $y$ both have neighbors in $C$ and item~\ref{pmc_char_2} is true for $K$ in $G=G'+a$.

    Therefore, since $K$ is not a PMC of $G$, it means that there exists a full component for $K$ in $G=G'+a$.
    Necessarily this component contains $a$, otherwise it would also be a connected component of $G'\setminus K$, and by hypothesis it is not full.
    In this case, we prove using Theorem~\ref{thm:characterization_pmc} that $K\cup\{a\}$ is a PMC of $G$.
    \begin{enumerate}\itemsep0em
        \item[\ref{pmc_char_1}] The connected components of $G\setminus (K\cup\{a\})$ are the same as those of $G'\setminus K$, so for any connected component $C$ there exist elements of $K\subseteq K\cup\{a\}$ that do not have neighbors in $C$.
        Consequently, there are no full components for $K\cup\{a\}$ in $G$.
        \item[\ref{pmc_char_2}] Let $x$ and $y$ be two vertices of $K\cup\{a\}$ such that $xy\notin E$.
        If $x\neq a$ and $y\neq a$, then condition~\ref{pmc_char_2} is satisfied by $x$ and $y$, since the connected components of $G\setminus (K\cup\{a\})$ are the same as those of $G'\setminus K$.
        Otherwise, we can assume $x=a$ and $ay\notin E(G)$.
        Since there exists a full component for $K$ in $G=G'+a$, in particular there exists a connected component of $G\setminus(K\cup\{a\})$ such that both $a$ and $y$ have neighbors in this component.
        Therefore, condition~\ref{pmc_char_2} is satisfied.
    \end{enumerate}
    Both conditions are satisfied, so by Theorem~\ref{thm:characterization_pmc}, $K\cup\{a\}$ is a PMC of $G$ and the proposition is proved.
    Moreover, since the PMCs of a graph are incomparable sets, we are sure that these are mutually exclusive: $K$ and $K\cup\{a\}$ cannot be both PMCs of the same graph.
    \qed
\end{proof}

\begin{algorithm}[htb]
  \DontPrintSemicolon
  \KwIn{A graph $G=(V,E)$}
  \KwOut{The set of all potential maximal cliques $\Pi_{n}$}
  \SetKwProg{myproc}{Function}{}{}
  \SetKwFunction{enum}{ONE\_MORE\_VERTEX}
  \SetKwFunction{ispmc}{IsPMC}
  \SetKw{and}{and}

  $\Pi_1 \gets \{v_1\}$\;
  \lFor{$i = 2, \dots, |V| $}{
    $\Pi_i \gets \varnothing$;
    \enum{$G, i$}
  }
  \Return $\Pi_n$\;

    \BlankLine

  \myproc{\enum{$G, i$}}{
    \ForAll(){$\pi \in \Pi_{i-1}$  }{
      \lIf{\ispmc{$\pi \cup \{v_i\}, G_i$}}{ADD $\pi \cup \{v_i\}$ to $\Pi_{i}$ }
      \lIf{\ispmc{$\pi, G_i$}}{ADD $\pi$ to $\Pi_{i}$ }
    }

    \ForAll(){$S \in \Delta_i$}{
      \lIf(){\ispmc{$S \cup \{v_i\}, G_i$}}{ ADD $S \cup \{v_i\}$ to $\Pi_{i}$ }
      
      \If(){$v_i \notin S$ \and $S \notin \Delta_{i-1}$}{
        \ForAll{$T \in \Delta_i$}{
          \ForAll{$C$ full component associated to $S$ in $G_i$}{
            \lIf{\ispmc{$S \cup (T \cap C), G_{i}$}}{
             ADD $S \cup (T \cap C)$ to $\Pi_i$
             }
           }
        }
      }
    }
  }
  \caption{Potential maximal cliques enumeration by \cite{bouchitte2002listing}.}
  \label{alg:original_bandt}
\end{algorithm}

\paragraph{Potential maximal cliques generation and minimal separators.}
In the current best known algorithm for enumerating the PMCs, it is crucial to be able to pass quickly through the list of all minimal separators, in what we called subroutine \texttt{GEN}.
The complexity status of the minimal separators enumeration problem has evolved since the introduction of PMCs.
In 2000, Berry et al. \cite{berry2000generating} provided a \emph{polynomial delay} algorithm for the minimal separators enumeration, meaning that the running time needed between two consecutive outputs is polynomial in the input size only.
However, it needs exponential space.
In 2010, Takata \cite{takata2010space} managed to enumerate the minimal $(a,b)$-separators for one pair $(a,b)$ with polynomial delay in polynomial space.
He also gave an output-polynomial algorithm in polynomial space for minimal $(a,b)$-separators for all pairs $(a,b)$.
More recently (WEPA 2019), Bergougnoux, Kanté and Wasa \cite{bergougnoux-disjunctive} presented an algorithm enumerating the minimal $(a,b)$-separators for all pairs $(a,b)$ in polynomial space, with (amortized) polynomial delay.
From a theoretical point of view, it is the currently best known algorithm for the enumeration of all minimal separators.

\paragraph{Original algorithm.} Algorithm~\ref{alg:original_bandt} shows the original strategy proposed in \cite{bouchitte2002listing} for PMC enumeration, based on an iterative argument: {generate the set $\Pi_i$ of PMCs of $G_i$ and keep it in memory to use it to compute $\Pi_{i+1}$ at the next step.}
New potential maximal cliques can be found in two ways, either by expanding an existing PMC from the previous step, or by extending a minimal separator.
The algorithm follows both roads, first try to expand existing cliques, then generate new ones from minimal separators.
This idea can be implemented by storing the family of sets $\mathcal{P} = \left\{ \Pi_i \,\mid\, i \in [1, n]\right\}$.
{The algorithm stops at the end of step $n$, when the set $\Pi_n$ containing all the PMCs of $G_n=G$ has been generated, and returns the whole set of solutions at the end of its execution.}
The strategy is summarized in \cite[Theorem 23]{bouchitte2002listing}, by the function \texttt{One\_More\_Vertex}, which is called for all $i = 1, \dots, n$.
The total time complexity is $O(n^2m|\Delta_G|^2)$.
However, the additional space required by the algorithm is $O(|\mathcal P|) = O(\sum_{i=1}^{n} |\Pi_i|) = O(n |\Pi_n|)$, {a bound that is clearly exponential in $n$} because all the solutions for all $G_i$ must be stored.

\section{Duplication Avoidance in the B\&T Algorithm}\label{sec:non-duplication}

Since our goal is to have a polynomial space algorithm, {our first task is to} rethink the \texttt{One\_More\_Vertex} strategy in a way such that it does not output the same solution twice.

Algorithm~\ref{alg:nondup_onemore_vertex} still stores the family of sets $\mathcal{P}$, but contains additional checks for duplication avoidance, so that we never try to add the same potential maximal clique twice to the same set. 
In particular, for a set $D^*$ generated at line 15 from the sets $S^*$, $T^*$ and $C^*$, the \texttt{Not\_Yet\_Seen}($D^*$)  check works as follows: the nested loops on $S$, $T$ and $C$ are run again to generate all the possible sets $D$ until $D^*$ is found for the first time.
If $D^*$ is found from $S^*$, $T^*$ and $C^*$ during this procedure, then \texttt{Not\_Yet\_Seen}($D^*$) returns $true$, otherwise it returns $false$.
\begin{algorithm}[htb]
  \DontPrintSemicolon
  \KwIn{A graph $G=(V,E)$}
  \KwOut{The set of all potential maximal cliques $\Pi_{n}$}
  \SetKwProg{myproc}{Function}{}{}
  \SetKwFunction{enum}{NONDUP\_ONE\_MORE\_VERTEX}
  \SetKwFunction{nys}{Not\_Yet\_Seen}
  \SetKwFunction{ispmc}{IsPMC}
  \SetKw{add}{append}
  \SetKw{and}{and}

  $\Pi_1 \gets \{v_1\}$\;
  \lFor{$i = 2, \dots, |V| $\label{alg:nondup_onemore_vertex:line:outer_for}}{
    $\Pi_i \gets \varnothing$;
    \enum{$G, i$}
  }
  \Return $\Pi_n$\;

    \BlankLine

  \myproc{\enum{$G, i$}}{
    \ForAll(\tcp*[h]{EXT(i)}){$\pi \in \Pi_{i-1}$ \label{alg:nondup_onemore_vertex:line:forall-previous} }{
      \lIf{\ispmc{$\pi \cup \{v_i\}, G_i$}}{\add{} $\pi \cup \{v_i\}$ to $\Pi_{i}$ \label{alg:nondup_onemore_vertex:line:already_added_1}}
      \lIf{\ispmc{$\pi, G_i$}}{\add{} $\pi$ to $\Pi_{i}$ \label{alg:nondup_onemore_vertex:line:already_added_2}}
    }

    \ForAll(\tcp*[h]{GEN(i)}){$S \in \Delta_i$ \label{alg:nondup_onemore_vertex:line:loop-separator}}{
      \If(){\ispmc{$S \cup \{v_i\}, G_i$} \label{alg:nondup_onemore_vertex:line:before_first_if}}{
      \lIf(){\textcolor{OliveGreen}{\textbf{!}\,\ispmc{$S, G_{i-1}$}}  \label{alg:nondup_onemore_vertex:line:first_if}}
        {\add{} $S \cup \{v_i\}$ to $\Pi_{i}$ } \label{alg:nondup_onemore_vertex:line:corresponding1}
        }
      
      \If(){$v_i \notin S$ \and{} $S \notin \Delta_{i-1}$\label{alg:nondup_onemore_vertex:line:if_vi_not_S}}{
        \ForAll{$T \in \Delta_i$ \label{alg:nondup_onemore_vertex:line:forallT}}{
          \ForAll{$C$ full component associated to $S$ in $G_i$ \label{alg:nondup_onemore_vertex:line:forallC}}{
            \If{\ispmc{$S \cup (T \cap C), G_{i}$}\label{alg:nondup_onemore_vertex:line:before_corresponding2}}{
            \begingroup
            \color{OliveGreen}
            $D \gets S \cup (T \cap C)$\; \label{alg:nondup_onemore_vertex:line:corresponding2}
            \If{\nys{$D$}\label{alg:nondup_onemore_vertex:line:check4}}{ 
            \If{\textbf{!}\,\ispmc{$D, G_{i-1}$}\label{alg:nondup_onemore_vertex:line:DGi-1}}{
                \If{$T \cap C \neq \{v_i\}$ \label{alg:nondup_onemore_vertex:line:check3}}{
                      \eIf{$v_i \in T \cap C$\label{alg:nondup_onemore_vertex:line:check5-1}}{
                        \If{\textbf{!}\,\ispmc{$D \setminus \{v_i\}, G_{i-1}$} \and{} $D \setminus \{v_i\} \notin \Delta_i$\label{alg:nondup_onemore_vertex:line:D_check_forall}}{\add{} $D$ to $\Pi_i$} \label{alg:nondup_onemore_vertex:line:addD}
                      }{
                        \add{} $D$ to $\Pi_i$\;\label{alg:nondup_onemore_vertex:line:last}
                      }
                  }
                }
                }
                \endgroup
                } 
            }
          }
        }
      }
    }
  \caption{PMC enumeration without duplication based on~\cite{bouchitte2002listing}.}
  \label{alg:nondup_onemore_vertex}
\end{algorithm}


\begin{lemma}
    Algorithm~\ref{alg:nondup_onemore_vertex} is correct and produces the same potential maximal cliques as the original Bouchitté and Todinca's algorithm \cite{bouchitte2002listing}, without duplicates.
\end{lemma}

\begin{proof}
Algorithm~\ref{alg:nondup_onemore_vertex} differs from the initial algorithm by \cite{bouchitte2002listing} only in the highlighted parts.
The changes consist in introducing additional checks before adding some given PMC to $\Pi_i$.
Specifically, we add the following checks:
\begin{enumerate}[label=$(\roman*)$]
    \item\label{lemma1:item:0} \nys{$D$} at line~\ref{alg:nondup_onemore_vertex:line:check4};
    \item\label{lemma1:item:1} \textbf{!}\,\ispmc{$S, G_{i-1}$} at line~\ref{alg:nondup_onemore_vertex:line:first_if};
    \item\label{lemma1:item:2} \textbf{!}\,\ispmc{$D, G_{i-1}$} at line~\ref{alg:nondup_onemore_vertex:line:DGi-1};
    \item\label{lemma1:item:3} $T\cap C \neq \{v_i\}$ at line~\ref{alg:nondup_onemore_vertex:line:check3};
    \item\label{lemma1:item:4} \textbf{!}\,\ispmc{$D\setminus \{v_i\}, G_{i-1}$} \textbf{and} $D \setminus \{v_i\} \not \in \Delta_i$ at line~\ref{alg:nondup_onemore_vertex:line:D_check_forall}.
\end{enumerate}
The \emph{PMC corresponding to a check} is the PMC that is not added to $\Pi_i$ when such check fails. 
Specifically, the PMC corresponding to check~\ref{lemma1:item:1} is $S \cup \{v_i\}$ (line~\ref{alg:nondup_onemore_vertex:line:corresponding1}), while for all other checks the corresponding PMC is set $D$ defined at line~\ref{alg:nondup_onemore_vertex:line:corresponding2} (see lines~\ref{alg:nondup_onemore_vertex:line:addD} and~\ref{alg:nondup_onemore_vertex:line:last}).  
As the underlying enumeration strategy is unchanged, the correctness of our algorithm follow from these two statements: 
\begin{enumerate}\itemsep0em
    \item check~\ref{lemma1:item:0} fails if and only if $D$ was previously processed by another choice of $S,T,C$ over the three for loops.
    The validity of this statement directly follows from the definition of Function~\nys{}.
    \item\label{it:prooflem1-2} a check among \ref{lemma1:item:1}-\ref{lemma1:item:4} fails in the  call of \texttt{Nondup\_One\_More\_Vertex}($G,i$) if and only if the corresponding potential maximal clique already belongs to $\Pi_i$ at that moment of computation. 
\end{enumerate} 
These statements also guarantee that each solution is inserted only once into $\Pi_n$ during the execution of Algorithm~\ref{alg:nondup_onemore_vertex}, that is to say, no duplicate solution is processed.
In what follows we prove item~\ref{it:prooflem1-2} by analyzing checks \ref{lemma1:item:1}-\ref{lemma1:item:4} separately.

First, assume that check \ref{lemma1:item:1} is reached and fails: this happens if and only if $S \cup \{v_i\}$ is a potential maximal clique of $G_i$, and $S \in \Pi_{i-1}$, which is true if and only if the for loop at line~\ref{alg:nondup_onemore_vertex:line:forall-previous} considered $\pi =S$ at some point, and the check at line~\ref{alg:nondup_onemore_vertex:line:already_added_1} was successful, meaning that $S \cup \{v_i\}$ was added to $\Pi_i$ at that time.

Let us now consider checks \ref{lemma1:item:2}-\ref{lemma1:item:4}, all corresponding to the same potential maximal clique $D$.
For these checks we have a shared setting, i.e., checks \ref{lemma1:item:2}-\ref{lemma1:item:4} happen when $S \in \Delta_i \setminus \Delta_{i-1}$ such that $v_i \not \in S$; $T \in \Delta_i$; $C$ is a full component associated to $S$ in $G_i$, and finally $D = S\cup (T \cap C)$ is a PMC of $G_i$.
We refer to this specific setting as the \emph{common checks setting}, and it will serve as a set of hypotheses for the rest of the proof. 
Notice that this setting is the same set of checks performed by Algorithm~\ref{alg:original_bandt} to determine if $D$ is a PMC in $G_i$.

In the common checks setting, condition \ref{lemma1:item:2} is reached and fails if and only if \ispmc{$D, G_{i-1}$}, i.e., $D \in \Pi_{i-1}$. This happens if and only if $D$ was considered as $\pi$ during the \texttt{forall} loop at line~\ref{alg:nondup_onemore_vertex:line:forall-previous} and the check at line~\ref{alg:nondup_onemore_vertex:line:already_added_2} was successful (so that $D$ is a potential maximal clique of $G_i$), thus if and only if $D$ was added to $\Pi_i$ at that time.

Consider now check \ref{lemma1:item:3}: assuming the common checks setting, we reach and fail this check if and only if \ref{lemma1:item:2} succeeds and $T \cap C = \{v_i\}$, which is equivalent to $D = S \cup \{v_i\}$. Thus, this check fails if and only if $D$ was already added at line~\ref{alg:nondup_onemore_vertex:line:already_added_1} or at line~\ref{alg:nondup_onemore_vertex:line:corresponding1} (according to whether $S \in \Pi_{i-1}$ or not).

Finally, check \ref{lemma1:item:4} is reached, in the common checks setting, whenever \ref{lemma1:item:2} and \ref{lemma1:item:3} succeed, and $v_i \in T \cap C$.
Thus for \ref{lemma1:item:4} to fail it is either that $D \setminus \{v_i\}$ is a PMC of $G_{i-1}$ or $D \setminus \{v_i\}$ is a minimal separator of $G_i$. 
We have that $D \setminus \{v_i\}$ is a PMC of $G_{i-1}$ if and only if the algorithm already processed the same set $D$ on line~\ref{alg:nondup_onemore_vertex:line:already_added_2}, thus we do not add it again to $\Pi_i$.
As for the second part, we only need to consider what happens when $D \setminus \{v_i\}$ is not a PMC of $G_{i-1}$, but $D \setminus \{v_i\}$ is a minimal separator of $G_i$. This happens if and only if $S= D \setminus \{v_i\}$ was already processed during the \texttt{forall} loop of line~\ref{alg:nondup_onemore_vertex:line:loop-separator}, and was added to $\Pi_i$ at line~\ref{alg:nondup_onemore_vertex:line:first_if}.
\qed
    
\end{proof}
\begin{lemma}
    Algorithm~\ref{alg:nondup_onemore_vertex} has output polynomial time complexity of $O(n^8 |\Pi_G|^4)$.
\end{lemma}
\begin{proof}
Before describing the complexity of the algorithm, note that the \texttt{IsPMC} check function can be implemented to run in $O(nm)$ time, using Theorem~\ref{thm:characterization_pmc} and \cite[Corollary 12]{bouchitte2002listing}. Additionally, 
{we assume that we can generate the set $\Delta_i$ with $O(n^3m)$ delay and polynomial space, using the algorithm from \cite{bergougnoux-disjunctive}.}
Finally, in what follows we assume that the sets $\Pi_i$, $\Delta_i$ and each $S \in \Delta_i$ are implemented as linked lists, so that we can append a new potential maximal clique to them in $O(1)$ time, while membership checks require linear time. 

We start by analyzing the \texttt{Nondup\_One\_More\_Vertex} function: this can be calculated by summing the complexity of the first \texttt{forall} loop on line~\ref{alg:nondup_onemore_vertex:line:forall-previous} and the second \texttt{forall} loop on line~\ref{alg:nondup_onemore_vertex:line:loop-separator}{, plus the $O(n^3m|\Delta_i|)$ time required to compute $|\Delta_i|$}.
The first loop, performs two \texttt{IsPMC} checks per element of $\Pi_{i-1}$.
Thus, its total running time is $O(|\Pi_{i-1}|nm) = O(nm|\Pi_G|)$. 

{The costful part is the second \texttt{forall} loop.}
It executes $|\Delta_i|$ iterations, during each of which we (a) perform two \texttt{IsPMC} checks ($O(nm)$ time) at lines~\ref{alg:nondup_onemore_vertex:line:before_first_if}-\ref{alg:nondup_onemore_vertex:line:first_if}; (b) check if $v_i$ belongs to $S$ and if $S$ belongs to $\Delta_{i-1}$, by scanning respectively $S$ and $|\Delta_{i-1}|$ in $O(|S|) = O(n)$ and $O(|\Delta_{i-1}|)$; (c) perform the \texttt{forall} loop at line~\ref{alg:nondup_onemore_vertex:line:forallT}.
The complexity of the second loop is therefore $O(|\Delta_i|(nm + |\Delta_{i-1}| + n + C_{12}))$, where $C_{12}$ is the complexity of the loop at line~\ref{alg:nondup_onemore_vertex:line:forallT}, which we now analyze.
The loop over $T \in \Delta_i$ counts $|\Delta_i|$ iterations, each of which costs $O(n \cdot (n + nm + n^2|\Delta_i|^2))$. Indeed, we iterate over all full components $C$ associated to $S$ in $G_i$, which can be $n$ in the worst case (when each vertex is a separate full connected component); for each of these we need to compute $S \cup (T \cap C)$ and later check if $v_i\in T\cap C$, which can both be done in $O(n)$ time, then perform up to three \texttt{IsPMC} calls in $O(nm)$ time. Finally, we need further $O(n^2|\Delta_i|^2)$ time for computing \texttt{Not\_Yet\_Seen}, as implemented with the loop restart mentioned above, and another $O(|\Delta_i|)$ to check whether $D\setminus v_i$ belongs to $\Delta_i$ at line~\ref{alg:nondup_onemore_vertex:line:D_check_forall}. 
Putting everything together we have:
    \begin{align*}
        &O(
        \underbracket{|\Delta_i| \phantspace}_{\text{\ref{alg:nondup_onemore_vertex:line:loop-separator}}} 
        (\underbracket{nm \phantspace}_{\text{\ref{alg:nondup_onemore_vertex:line:before_first_if}-\ref{alg:nondup_onemore_vertex:line:first_if}}}
        + \underbracket{|\Delta_i|+n\phantspace}_{\text{\ref{alg:nondup_onemore_vertex:line:if_vi_not_S}}}
        \underbracket{|\Delta_i|\phantspace}_{\text{\ref{alg:nondup_onemore_vertex:line:forallT}}}(\underbracket{n\phantspace}_{\text{\ref{alg:nondup_onemore_vertex:line:forallC}}}
        (\underbracket{nm + n\phantspace}_{\text{\ref{alg:nondup_onemore_vertex:line:before_corresponding2}-\ref{alg:nondup_onemore_vertex:line:corresponding2}}}
        + \underbracket{n^2|\Delta_i^2| \phantspace}_{\text{\ref{alg:nondup_onemore_vertex:line:check4}}}
        + \underbracket{nm\phantspace}_{\text{\ref{alg:nondup_onemore_vertex:line:DGi-1}}}
        + \underbracket{\hphantom{n} n \phantspace \hphantom{n}}_{\text{\ref{alg:nondup_onemore_vertex:line:check3}-\ref{alg:nondup_onemore_vertex:line:check5-1}}}
        + \underbracket{nm + |\Delta_i|\phantspace}_{\text{\ref{alg:nondup_onemore_vertex:line:D_check_forall}}})))) \\
        =\ &O(
        |\Delta_i|
        (
        n^3 |\Delta_i|^3 + n^2m |\Delta_i|
        )
        )  
        = O(
        n^2 |\Delta_i|^2
        (
        n |\Delta_i|^2+ m
        )
        ).
    \end{align*}%
    This is the cost of a call to the {second loop of Function~}\texttt{Nondup\_One\_More\_Vertex} with fixed $i$.
    Summing this cost {with the cost of the first \texttt{forall} loop} and of the computation of $\Delta_i$ for all the $n$ calls to the function from line~\ref{alg:nondup_onemore_vertex:line:outer_for} we obtain the overall time complexity for Algorithm~\ref{alg:nondup_onemore_vertex}: 
    $O(n^2m|\Pi_G| + n^4m|\Delta_G|+ n^3 |\Delta_G|^2 \cdot(n|\Delta_G|^2 + m))=$ $O(n^2m|\Pi_G| + n^4m|\Delta_G| + n^3 |\Delta_G|^2 \cdot(n|\Delta_G|^2 + m))$.
    
    The total complexity of Algorithm~\ref{alg:nondup_onemore_vertex} is polynomial in the number of PMCs of the graph.
    To highlight this, we use the known inequality $|\Pi_G| \geq |\Delta_G|/n$ \cite{bouchitte2002listing}.
    Thus we obtain the final output-sensitive complexity of 
    \begin{align*}
        &O(n^2m|\Pi_G| + n^4m|\Delta_G| + n^3 |\Delta_G|^2  (n |\Delta_G|^2 + m)) \\ 
         =\ & O(n^2m|\Pi_G| + n^5m|\Pi_G| + n^3(n|\Pi_G|)^2 \cdot ((n|\Pi_G|)^2 n + m)) \\
         =\ & O(n^5m|\Pi_G| + n^5m|\Pi_G| + n^5 |\Pi_G|^2 \cdot (n^3 |\Pi_G|^2 + m))  = O(n^8 |\Pi_G|^4).\qquad \qed
    \end{align*}
\end{proof}

\section{Polynomial Space Algorithm}\label{sec:polyspace}

Algorithm~\ref{alg:nondup_onemore_vertex} has output-sensitive time complexity and lists all the potential maximal cliques without duplicates, but it still uses more than polynomial space because it has to store the sets of solutions $\Pi_i$ before returning $\Pi_n$.
Therefore, in this section we show how to adapt the algorithm to output solutions as soon as they are found, without having to store $\Pi_i$ for duplicate detection.
The key idea is to change the way in which we traverse the solution space from breadth-first to depth-first.
Said otherwise, once we get a PMC $K$ in $G_i$, we immediately extend it to a PMC of $G_n$, which is always possible according to Proposition~\ref{prop:pmc-extension}, before looking for a new one.
This idea is summarized in Algorithm~\ref{alg:pmc_dfs}.

\begin{algorithm}[htb]
  \DontPrintSemicolon
  \KwIn{An integer $i$ corresponding to $v_i \in V$, a PMC $\pi$ in $G_{i-1}$ ($G$ is implicit)}
  \SetKwProg{myproc}{Function}{}{}
  \SetKwFunction{rec}{EXT}
  \SetKwFunction{gen}{GEN}
  \SetKwFunction{nys}{Not\_Yet\_Seen}
  \SetKwFunction{ispmc}{IsPMC}
  \SetKw{and}{and}

  \For{$i = 1, \dots, |V|$}{
      \lForAll{$\pi' \in $ \gen{$G, i$}}{
        \rec{$i+1, \pi'$}
      }
    }

    \myproc{\rec{$i, \pi$} }{
      \lIf{$i=n$}{ \textbf{output} $\pi$}

      \lIf{\ispmc{$\pi \cup \{v_i\}, G_i$}}{
        \rec{$i+1, \pi \cup \{v_i\}$}
      }
      \lElseIf{\ispmc{$\pi, G_i$}}{
        \rec{$i+1, \pi$}
      }
    }
    
  \myproc{\gen{$G, i$}}{
      \lIf(){$i=1$}{
        \textbf{yield} $\{v_1\}$ \and{} \Return
      }
      
    \ForAll(\tcp*[h]{$S$ is a minimal separator in $G_i$}){$S \in \Delta_i$}{
      \If(){{\textbf{!}\,\ispmc{$S, G_{i-1}$}} \and{} \ispmc{$S \cup \{v_i\}, G_i$}}
        {\textbf{yield} $S\cup \{v_i\}$\label{alg:pmc_dfs:line:cupvi}}
      
      \If(){$v_i \notin S$ \and{} $S \notin \Delta_{i-1}$}{
        \ForAll{$T \in \Delta_i$}{
          \ForAll{$C$ full component associated to $S$ in $G_i$}{
            \If{\ispmc{$S \cup (T \cap C), G_i$}}{
                $D \gets S \cup (T \cap C)$\;
                \If{\nys{$D$}\label{alg:pmc_dfs:line:nys}}{
                    \If{\textbf{!}\,\ispmc{$D, G_{i-1}$}\label{alg:pmc_dfs:line:checkD}}{
                        \If{$T \cap C \neq \{v_i\}$}{
                            \eIf{$v_i \in T \cap C$}{
                                \If{\textbf{!}\,\ispmc{$D \setminus \{v_i\}, G_{i-1}$\label{alg:pmc_dfs:line:checkD-minus-vi}} \and{} $D \setminus \{v_i\} \notin \Delta_i$}{\textbf{yield} $D$} 
                            }{ 
                                \textbf{yield} $D$\;
                            }
                        }
                    }
                }
            }
        }
    }
}
}
}

    \caption{Depth-first algorithm for output-polynomial time, polynomial space potential maximal clique enumeration.}
    \label{alg:pmc_dfs}
\end{algorithm}
First, we split the \texttt{One\_More\_Vertex} function into two parts: the \texttt{EXT} routine extends a given potential maximal clique $\pi$ in $G_i$ to a PMC of $G$ by recurring on $i+1, i+2, \dots, n$.
The \texttt{EXT} procedure is presented here as a recursive function in order to highlight the key idea of a depth-first traversal of the search space, but it is easy to rewrite it iteratively, saving the additional space required for handling recursion.
The second part of the original function is the \texttt{GEN} routine,
which generates all the new PMCs of the graph $G_i$, yielding them one by one to the \texttt{EXT} procedure.
In particular, we see the \texttt{GEN} function as a \emph{generator} of PMCs, i.e., the solutions are yielded during the execution while keeping the internal state\footnote{Consider \texttt{GEN}($G, i$) as an iterator over the set $\Pi_i$. The first call, for each $i$, will compute everything needed to output $\pi_{i,1} \in \Pi_i$. Then, when called with the same $i$, it will produce $\pi_{i,2} \in \Pi_i$ without recomputing everything from scratch and so on.\label{footnote}}.
This way the total running time of the function does not change, and we process each PMC immediately after it has been generated.

\begin{proposition}\label{prop:claim}
    Let $K$ be a potential maximal clique of $G_i$.
    If $K$ is not a PMC of $G_{i+1}$, then it cannot be a PMC of any $G_j$ with $j > i$. 
\end{proposition}
\begin{proof}
    If $K$ is a PMC of $G_i$ but not of $G_{i+1}$, then by Proposition~\ref{prop:pmc-extension} $K\cup\{v_{i+1}\}$ must be a PMC of $G_{i+1}$. So, $K$ is strictly included in a PMC of $G_{i+1}$. By repeatedly applying this reasoning, $K$ is also strictly included in some PMC of $G_j$ for any $j>i$. Therefore, since PMCs are not included in each other, $K$ cannot be a PMC of $G_j$ for any $j>i$.
    \qed
\end{proof}

\begin{corollary}\label{cor:prefix-filtered}
    Let $K_i$, $K_j$ be two PMCs of $G_i$ and $G_j$ respectively.
    Then, $K_i$ and $K_j$ cannot be extended to the same $K$ in $G$, unless $K_i \subset K_j$.
\end{corollary}
\begin{proof}
    Assume, for the sake of contradiction, that $K_i$ and $K_j$ get extended to the same $K$ in $G$ without being one the subset of the other.
    This extension implies that $K_i \subseteq K$ and $K_j \subseteq K$.
    Now consider, without loss of generality, the set $D = K_j \setminus K_i$: $D$ cannot be empty as $K_i \not\subset K_j$, and $D \subseteq K$, because $K_i$ and $K_j$ got expanded to the same $K$.

    However, by definition of set difference, vertices in $D$ do not belong to $K_i$ in $G_j$, so $K_i$ cannot be extended to the same $K$ as $K_j$, a contradiction.
    Thus we conclude that $K_i$ and $K_j$ can be extended to the same potential maximal clique $K$ in $G$ only if $K_i \subset K_j$.
    \qed
\end{proof}

\begin{lemma}
    Algorithm~\ref{alg:pmc_dfs} outputs all and only the potential maximal cliques of $G$, without duplication.
\end{lemma}

\begin{proof}
First, as Algorithm~\ref{alg:pmc_dfs} is a rearrangement of Algorithm~\ref{alg:nondup_onemore_vertex} in a depth-first fashion, it must also output every PMC at least once.
It remains to prove that after splitting the algorithm in two distinct functions, every PMC of $G$ is produced by Algorithm~\ref{alg:pmc_dfs} in exactly one way.
Note that line~\ref{alg:pmc_dfs:line:cupvi} will never output a duplicated PMC, since vertex $v_i$ is different $\forall\,i \leq n$.
Now suppose that Algorithm~\ref{alg:pmc_dfs} outputs a duplicated PMC $K$, we have two cases.
\begin{enumerate}[label=(\alph*)]
    \item The same set $D$ is found twice by \texttt{GEN} at steps $i$ and $j$: it must be that $i \neq j$ (we assume $i < j$ w.l.o.g.), i.e., the duplicated $D$ comes from different calls \texttt{GEN}($i$) and \texttt{GEN}($j$).
    This is because of the check at line~\ref{alg:pmc_dfs:line:nys} of Algorithm~\ref{alg:pmc_dfs} that explicitly prevents this.
    Thus, $D$ is a PMC in each step between $i$ and $j$ by Proposition~\ref{prop:claim}, including step $j-1$, so it will be filtered by line~\ref{alg:pmc_dfs:line:checkD}.
    \item $K$ is output twice by \texttt{EXT}: in this case $K$ comes from two different sets $K_i$ and $K_j$ produced by \texttt{GEN} at different steps $i$ and $j$ (with $i < j$).
    By Corollary~\ref{cor:prefix-filtered}, $K_i\subseteq K_j$.
    Since $K_i$ and $K_j$ have been produced by \texttt{GEN}, they are PMCs of $G_i$ and $G_j$ respectively.
    Then, as a consequence of Proposition~\ref{prop:pmc-extension}, for any $i<i'<j$ there exists $K_i\subseteq K_{i'}\subseteq K_j$ such that $K_{i'}$ is a PMC of $G_{i'}$.
    In particular, if $i'=j-1$, $K_j\setminus{v_j}$ (that is, $K_j$ if $v_j\notin K_j$) is a PMC of $G_{j-1}$ and is therefore filtered at line~\ref{alg:pmc_dfs:line:checkD} or~\ref{alg:pmc_dfs:line:checkD-minus-vi}.
\end{enumerate}
\vspace{-.29em}
Thus we conclude that Algorithm~\ref{alg:pmc_dfs} cannot output duplicated solutions.
\qed \end{proof}

\begin{lemma}
    Algorithm~\ref{alg:pmc_dfs} has output-polynomial time complexity and uses polynomial space. Namely, it uses $O(n^9 m^2 |\Pi_G|^4 )$ time and $O(n^3)$ space.
\end{lemma}
\begin{proof}
We first recall that, at each step $i$, we can enumerate separators of $\Delta_i$ with $O(n^3m)$ delay by \cite{bergougnoux-disjunctive}.
We start by analyzing the time complexity of the body of the outermost for loop of line 1.
    We see the function \texttt{GEN} as an iterator\footref{footnote}, so that it yields new PMCs during its execution.
    The time complexity of \texttt{GEN}$(G, i)$ for a fixed $i$ is $O(
        n^3 m |\Delta_i| + 
        n^2 |\Delta_i|^2
        \cdot
        (
        n |\Delta_i|^2+ m
        )
        )$ (same as Algorithm~\ref{alg:nondup_onemore_vertex}).
    The corresponding \texttt{EXT} call has maximum depth of $O(n)$ and a cost of $O(nm)$ due to the \texttt{IsPMC} checks, yielding $O(n^2m)$ time worst case and returning a new solution, by Proposition~\ref{prop:pmc-extension}.
    Thus, for a fixed $i$, the total cost of step $i$ is $O(n^4m^2|\Delta_i|^4)$.
    Since $|\Delta_i| \leq |\Delta_G|$ for all $i$, the total cost of Algorithm~\ref{alg:pmc_dfs} is $O(n^5m^2|\Delta_G|^4)$.
    As $|\Pi_G| \geq |\Delta_G|/n$ we obtain the final output polynomial complexity of $O(n^9 m^2 |\Pi_G|^4 )$.

    The space usage is $O(n^3+m) = O(n^3)$ due to the space complexity of the listing algorithm for minimal separators~\cite{bergougnoux-disjunctive}.
    Note that we do not explicitly store any solution, as we immediately output it at the end of the \texttt{EXT} computation.
    \qed
\end{proof}

\clearpage

\section{Conclusions}
This paper shows that potential maximal cliques can be enumerated in output-sensitive time using only polynomial space, rather than exponential space as required by existing approaches. While the complexity of our algorithm is still significant, this approach opens the way for the development of practical enumeration algorithms for PMCs, which would lead in turn to advancements on related graph problems such as treewidth decomposition and maximum-weight independent sets.

\clearpage
\bibliographystyle{splncs04}
\bibliography{bibliography}

\end{document}